\documentclass[aps,prl,reprint,onecolumn,superscriptaddress,longbibliography,amsmath,amssymb]{revtex4-2}
\usepackage[T1]{fontenc}
\usepackage[utf8]{inputenc}
\usepackage{amsthm,mathtools}
\usepackage{microtype}
\usepackage{booktabs}
\usepackage{graphicx}
\usepackage{tikz}
\usetikzlibrary{arrows.meta,positioning,calc}
\usepackage{hyperref}
\usepackage[nameinlink,noabbrev]{cleveref}
\hypersetup{hidelinks,
 pdftitle={A Simple Quantum Linear-System Solver via Dissipation},
 pdfauthor={Zhongxia Shang}}
\allowdisplaybreaks[2]
\newtheorem{theorem}{Theorem}

\newtheorem{proposition}[theorem]{Proposition}
\newtheorem{corollary}[theorem]{Corollary}
\theoremstyle{remark}

\DeclareMathOperator{\Tr}{Tr}
\DeclareMathOperator{\spanop}{span}

\DeclareMathOperator{\diag}{diag}
\newcommand{\ket}[1]{\lvert #1\rangle}
\newcommand{\bra}[1]{\langle #1\rvert}
\newcommand{\braket}[2]{\langle #1\vert #2\rangle}
\newcommand{\uket}[1]{\lvert #1\rangle\!\rangle}

\newcommand{\norm}[1]{\left\lVert #1\right\rVert}
\newcommand{\abs}[1]{\left\lvert #1\right\rvert}
\newcommand{\Lin}{\mathcal L}
\newcommand{\HS}{\mathcal H_S}
\newcommand{\HX}{\mathcal H_x}
\newcommand{\HR}{\mathcal H_{\mathrm{res}}}
\newcommand{\tmix}{t_{\mathrm{mix}}}
\newcommand{\eps}{\varepsilon}
\newcommand{\dd}{\mathrm d}
\newcommand{\ee}{\mathrm e}
\newcommand{\proofstep}[1]{\par\smallskip\noindent\textit{#1}\quad}

\begin{document}
\title{A Simple Quantum Linear-System Solver via Dissipation}
\author{Zhong-Xia Shang}
\email{zhongxia.shang@math.ku.dk}
\affiliation{Department of Mathematical Sciences, University of Copenhagen, Denmark}
\begin{abstract}
Dissipation has been recently demonstrated as a powerful primitive for designing quantum algorithms. We apply this viewpoint directly to linear-system solving, $Ax=b$. We construct a simple purely dissipative Lindbladian whose unique fixed point encodes the linear-system solution. We prove dimension-independent trace-distance mixing in $\Theta(\kappa^2\log(1/\eps))$ time. We then show how to run this Lindbladian on digital quantum computers via collective block encoding and Lindbladian simulation, resulting in an $O\!\left(\kappa^2\log(1/\eps)
             \frac{\log(\kappa/\eps)}{\log\log(\kappa/\eps)}\right)$ query complexity for both $U_A$, the block encoding of $A$, and $U_b$, the $|b\rangle$ state preparation unitary.
\end{abstract}
\maketitle

\section{Introduction}\label{sec:intro}
Uncontrolled dissipation and noise are usually regarded as obstacles to
quantum computation, motivating quantum error correction to protect
quantum information from their effects. However, suitably
engineered dissipation can itself serve as a computational resource:
time-independent, purely dissipative dynamics can perform universal
quantum computation, with the output encoded in a stationary
state~\cite{VWC}. Recent progress has extended this perspective to
concrete algorithmic tasks, including Gibbs-state
preparation~\cite{Chen2025ThermalSimulation,DingLiLin,RouzeFrancaAlhambra,RouzeFrancaAlhambraOptimal}
and ground-state preparation~\cite{DingChenLin}, solving linear ordinary
differential equations~\cite{ShangODE}, finding local energy minima~\cite{LocalMinima},
quantum phase estimation~\cite{ShangQPE}, and the approximation of
Hamiltonian dynamics~\cite{ShangFranca}. These approaches exploit designed
open-system dynamics to encode computational information either in
attractive stationary states or in transient evolution. Importantly,
such dynamics need not be realized by directly engineering a physical
environment: advances in Lindbladian simulation provide digital
implementations and observable-estimation methods on gate-based quantum
computers, under suitable access
to the generator's Hamiltonian and jump
operators~\cite{Kliesch,Sweke,ChildsLi,CleveWang,LiWang,CKBG,DingLiLinSimulation,Peng,YuLiZhaoYuan,KatoWadaItoYamamoto}, with the possibility of fast-forwarding~\cite{ShangQPE,ShangExpFF,GaoJiLiu}. This combination of
dissipative algorithm design and quantum simulation establishes
open-system dynamics as a computational primitive, rather than merely
a source of errors to be corrected.

Linear-system solving is a natural problem on which to develop this approach further. Given an invertible matrix $A$ and a normalized right-hand side $\ket b$, a quantum linear-system algorithm prepares a state proportional to $A^{-1}\ket b$, rather than explicitly outputting all solution entries. Harrow, Hassidim, and Lloyd established this state-preparation formulation~\cite{HHL}. Subsequent algorithms improved conditioning through variable-time amplitude amplification~\cite{Ambainis} and precision through Fourier and polynomial approximations~\cite{CKS,QSVT}. Complementary approaches use randomized eigenpath traversal~\cite{SubasiSommaOrsucci}, optimized adiabatic schedules~\cite{AnLin}, and polynomial eigenstate filtering~\cite{LinTong}; dense-matrix solvers also exploit coherent data access~\cite{WossnigZhaoPrakash}. Discrete adiabatic evolution and augmented-kernel reflection attain $O(\kappa\log(1/\eps))$ query complexity under normalized input access~\cite{Costa,Dalzell}. Queries to the state-preparation oracle were also recently optimized~\cite{LowSu}. These algorithms are primarily organized around coherent inversion or coherent state preparation. The question here is 
\begin{center}
\textit{Is it possible to solve linear systems from dissipation?}
\end{center}

We give a positive answer to this question. We construct a simple purely dissipative Lindbladian which has a unique attractive pure stationary state and dimension-independent worst-case mixing time $O(\kappa^2\log(1/\eps))$, with a matching lower bound for the specified rate-normalized family. The jump operator is based on encoding the solution as a residual kernel. We also show how to run this Lindbladian on quantum computers via the collective block encoding of jump operators and the Lindbladian simulation algorithm in~\cite{CKBG}.

\section{Setup and preliminaries}\label{sec:setup}
\subsection{Linear-system task and notation}
Let $A:\HX\to\HR$ be an invertible $N\times N$ matrix between two identified copies of $\mathbb C^N$, and let $\ket b\in\HR$ be normalized. Assume
\begin{equation}
 \norm A\le 1,\qquad \sigma_{\min}(A)\ge\kappa^{-1},\qquad \kappa\ge 1.
 \label{eq:assumptions}
\end{equation}
For vectors, $\norm{\cdot}$ denotes the Euclidean norm; for operators, it denotes the operator norm, while $\norm{\cdot}_1$ is the trace norm. Define
\begin{equation}
 \uket x=A^{-1}\ket b,\qquad
 s=\norm{\uket x},\qquad
 \ket x=\frac{\uket x}{s}.
 \label{eq:solution}
\end{equation}
Double kets and double bras denote vectors that need not be normalized; ordinary state kets denote normalized states. The bounds on $s$ follow separately from
\begin{equation}
 \begin{aligned}
 1&=\norm{\ket b}=\norm{A\uket x}\le\norm A\,\norm{\uket x}\le s,\\
 s&=\norm{A^{-1}\ket b}\le\norm{A^{-1}}\le\kappa.
 \end{aligned}
\end{equation}
Thus
\begin{equation}
 1\le s\le\kappa.
 \label{eq:solutionnorm}
\end{equation}
The goal is to prepare $\ket x\bra x$ to trace-distance error $\eps$. We use
\begin{equation}
 D(\rho,\sigma)=\frac12\norm{\rho-\sigma}_1.
 \label{eq:distance}
\end{equation}
The parameter $\kappa$ is a known inverse-singular-value bound under the chosen normalization. When $\norm A=1$ and the bound is attained, it is the usual condition number. We do not assume that $A$ is Hermitian or normal.

\subsection{Lindblad dynamics and mixing time}
A density matrix $\rho$ is positive semidefinite with $\Tr\rho=1$. A time-independent Lindblad generator~\cite{Lindblad1976,GKS1976} acts as
\begin{equation}
 \begin{aligned}
 \mathcal M(\rho)={}&-i[H,\rho]+\sum_\mu\left(L_\mu\rho L_\mu^\dagger
       -\frac12\{L_\mu^\dagger L_\mu,\rho\}\right),\\
 \{X,Y\}={}&XY+YX.
 \end{aligned}
 \label{eq:general-lindblad}
\end{equation}
Here $H$ is Hermitian and $L_\mu$ are jump operators on the same system Hilbert space. The equation $\dot\rho_t=\mathcal M(\rho_t)$ generates a completely positive, trace-preserving semigroup $\rho_t=\ee^{t\mathcal M}(\rho_0)$. We call the chosen representation \emph{purely dissipative} when $H=0$.

A stationary state satisfies $\mathcal M(\rho_*)=0$. It is globally attractive if every initial density matrix converges to it. For a stationary pure target $P_*$, define the mixing time
\begin{equation}
 \tmix(\eps)=\inf\left\{t\ge0:
  \sup_{\rho_0}D(\ee^{t\mathcal M}(\rho_0),P_*)\le\eps\right\}.
 \label{eq:mixdefinition}
\end{equation}

\subsection{Input access and resource conventions}
Assume access to a unitary block encoding $U_A$~\cite{LowChuang,QSVT} and a source-preparation unitary $U_b$ such that
\begin{equation}
 (\bra{0^a}\otimes I)U_A(\ket{0^a}\otimes I)=A,
 \qquad U_b\ket0=\ket b,
 \label{eq:oracles}
\end{equation}
including their inverses and controlled versions. An oracle query refers to a call to one of the supplied unitaries. The cost of constructing the oracles from classical data is separate and we refer to~\cite{QSVT,CampsSparse}.

\section{Algorithm construction}\label{sec:construction}
\subsection{Encoding the solution as a residual kernel}
The active system Hilbert space is
\begin{equation}
 \HS=\HX\oplus\spanop\{\ket r\},\qquad \dim\HS=N+1,
 \label{eq:space}
\end{equation}
where $\ket r$ is orthogonal to every solution-coordinate state. We identify $\HX$ with its embedded copy in $\HS$. The map $F:\HS\to\HR$ is defined by
\begin{equation}
 F\bigl(\uket v+\alpha\ket r\bigr)
   =A\uket v-\alpha\ket b.
 \label{eq:Faction}
\end{equation}
In the basis consisting of the solution coordinates followed by $\ket r$, this is $F=(A,-\ket b)$. The residual output space is not a second retained sector of the system; it supplies the labels used by the jump operators.

Since the reference state is orthogonal to the solution coordinates, every system vector has a unique decomposition $\uket v+\alpha\ket r$. The kernel condition is equivalent to
\begin{equation}
 A\uket v=\alpha\ket b
 \quad\Longleftrightarrow\quad
 \uket v=\alpha A^{-1}\ket b=\alpha\uket x.
\end{equation}
Thus every kernel vector is a multiple of $\uket x+\ket r$, and this vector has squared norm $s^2+1$. Consequently
\begin{equation}
 \ker F=\spanop\{\ket\psi\},\qquad
 \ket\psi=\frac{\uket x+\ket r}{\sqrt{s^2+1}}.
 \label{eq:target}
\end{equation}
Set
\begin{equation}
 \begin{gathered}
 P=\ket\psi\bra\psi,\qquad Q=I_S-P,\\
 R=\ket r\bra r,\qquad G=F^\dagger F.
 \end{gathered}
 \label{eq:projectors}
\end{equation}
The target overlap of the reset state is
\begin{equation}
 p=\Tr(PR)=\abs{\braket\psi r}^2=\frac1{s^2+1}.
 \label{eq:p}
\end{equation}
This overlap is positive, but is not assumed to be constant.

\paragraph{Residual gap and rate normalization.}
On the residual space,
\begin{equation}
 FF^\dagger=AA^\dagger+\ket b\bra b.
 \label{eq:FFdagger}
\end{equation}
For every normalized $\ket u\in\HR$,
\begin{equation}
 \begin{gathered}
 \bra uFF^\dagger\ket u
 =\norm{A^\dagger\ket u}^{\,2}+\abs{\braket b u}^{\,2},\\
 \kappa^{-2}\le\bra uFF^\dagger\ket u\le2.
 \end{gathered}
 \label{eq:residual-rayleigh}
\end{equation}
The lower bound uses $\sigma_{\min}(A^\dagger)=\sigma_{\min}(A)\ge\kappa^{-1}$; the upper bound uses $\norm{A^\dagger}\le1$ and $\norm{\ket b}=1$. Hence $FF^\dagger$ is positive definite, $F$ has full row rank, and its $N$ nonzero singular values lie in $[\kappa^{-1},\sqrt2]$. The corresponding eigenvalues of $G=F^\dagger F$ are their squares. Since $G$ is zero on $\ket\psi$ and these eigenvectors span its orthogonal complement,
\begin{equation}
\kappa^{-2}Q\le G\le 2Q.
 \label{eq:Ggap}
\end{equation}
The operator $G$ vanishes exactly on the target. Its bounded norm also fixes the rate normalization: the total jump rate defined below is at most two, independently of $N$ and $\kappa$. 

\subsection{The designed jumps and their stationary state}
Choose a computational basis $\{\ket j\}_{j=1}^N$ of $\HR$. Define operators on the active system by
\begin{equation}
J_j=\ket r\bra jF,\qquad j=1,\ldots,N.
 \label{eq:jumps}
\end{equation}
In coordinates,
\begin{equation}
 J_j=\ket r\left(\sum_{i=1}^N A_{ji}\bra i-b_j\bra r\right),
 \qquad b_j=\braket j b.
 \label{eq:jumprows}
\end{equation}
Thus each jump is built from a known row of the residual map and a known reset destination. 

\paragraph{Purely dissipative generator.}
The generator has no Hamiltonian term:
\begin{align}
 \Lin(\rho)
 &=\sum_{j=1}^N\left(J_j\rho J_j^\dagger
                -\frac12\{J_j^\dagger J_j,\rho\}\right) \notag\\
 &=\Tr(G\rho)R-\frac12\{G,\rho\}.
 \label{eq:lindblad}
\end{align}
To verify the two terms, use $\braket r r=1$ and $\sum_j\ket j\bra j=I_{\mathrm{res}}$:
\begin{align}
 \sum_jJ_j^\dagger J_j
 &=\sum_jF^\dagger\ket j\braket r r\bra jF\notag\\
 &=F^\dagger\left(\sum_j\ket j\bra j\right)F=G,\notag\\
 \sum_jJ_j\rho J_j^\dagger
 &=\sum_j\ket r\bigl(\bra jF\rho F^\dagger\ket j\bigr)\bra r\notag\\
 &=\Tr(F\rho F^\dagger)R=\Tr(G\rho)R.
 \label{eq:collective-identities}
\end{align}
In the second identity each parenthesized factor is a scalar. Cyclicity of the trace gives $\Tr(F\rho F^\dagger)=\Tr(F^\dagger F\rho)$. In particular,
$\Tr\Lin(\rho)=\Tr(G\rho)-\tfrac12\Tr(G\rho)-\tfrac12\Tr(\rho G)=0$.
The dynamics $\rho_t=\ee^{t\Lin}(\rho_0)$ is therefore a completely positive, trace-preserving semigroup. 

Since $F\ket\psi=0$, every jump satisfies $J_j\ket\psi=0$, and $G\ket\psi=F^\dagger F\ket\psi=0$. Hermiticity of $G$ also gives $\bra\psi G=0$. Thus $GP=PG=0$ and $\Tr(GP)=0$; substituting these three identities into \eqref{eq:lindblad} yields
\begin{equation}
 GP=PG=0,\qquad \Lin(P)=0.
 \label{eq:fixedpoint}
\end{equation}
The stationary state is the augmented solution. Uniqueness and attraction from every initial density matrix will be proved in Section~\ref{sec:mixing}.

\subsection{Physical picture and population flow}
For a normalized system state $\ket\phi=\uket v+\alpha\ket r$,
\begin{equation}
 J_j\ket\phi=\bigl(\bra j A\uket v-\alpha\braket j b\bigr)\ket r.
 \label{eq:jumpeffect}
\end{equation}
The amplitude for jump $j$ is a component of the residual. In the usual monitored-jump interpretation, its probability in a short time $\dd t$ is $\norm{J_j\ket\phi}^2\dd t+o(\dd t)$, and the state immediately after any detected jump is $\ket r$. The total jump intensity of a density matrix is
\begin{equation}
 j(t)=\Tr(G\rho_t).
 \label{eq:intensity}
\end{equation}
Monitoring is not required to implement the unconditional state preparation.

Between jumps, an unnormalized state evolves as
\begin{equation}
 \uket{\widetilde\phi_t}=\ee^{-tG/2}\ket{\phi_0}.
 \label{eq:nojump}
\end{equation}
Its norm squared is the probability of the corresponding jump-free interval. This is not a unitary evolution. The anticommutator in \eqref{eq:lindblad} supplies the loss associated with the jump probabilities, while the first term restores the trace by reinjecting population into $R$.

Let
\begin{equation}
 f(t)=\Tr(P\rho_t),\qquad q(t)=1-f(t)=\Tr(Q\rho_t).
 \label{eq:fq}
\end{equation}
Taking the target expectation of \eqref{eq:lindblad} gives
\begin{align}
 f'(t)&=j(t)\Tr(PR)
 -\tfrac12\Tr(PG\rho_t)-\tfrac12\Tr(P\rho_tG)\\
 &=p\,j(t)-\tfrac12\Tr(PG\rho_t)-\tfrac12\Tr(GP\rho_t)
 =p\,j(t).
\end{align}
The last two terms vanish because $PG=GP=0$. Since $G,\rho_t\ge0$, we have $j(t)\ge0$, and consequently
\begin{equation}
 f'(t)=p\,j(t),\qquad q'(t)=-p\,j(t)\le0.
 \label{eq:fidelity}
\end{equation}
Thus target population is nondecreasing in the unconditional density matrix. This is an ensemble statement; the fidelity of a state conditioned on a particular jump record need not be monotone.

\paragraph{Population bookkeeping.}
Let $\ket{v_j}$ be orthonormal non-target eigenvectors of $G$, with
\begin{equation}
 G\ket{v_j}=\lambda_j\ket{v_j},\qquad
 \lambda_j=\sigma_j(F)^2>0.
 \label{eq:modes}
\end{equation}
Write
\begin{equation}
 \ket r=\sqrt p\ket\psi+\sum_{j=1}^N r_j\ket{v_j},
 \qquad \sum_j\abs{r_j}^2=1-p,
 \label{eq:rdecomp}
\end{equation}
and let $a_j(t)=\bra{v_j}\rho_t\ket{v_j}$. The spectral decomposition $G=\sum_j\lambda_j\ket{v_j}\bra{v_j}$ gives $j(t)=\sum_j\lambda_j a_j(t)$. Furthermore,
\begin{align}
 a_j'(t)
 &=j(t)\bra{v_j}R\ket{v_j}
 -\tfrac12\bra{v_j}G\rho_t\ket{v_j}
 -\tfrac12\bra{v_j}\rho_tG\ket{v_j}\\
 &=\abs{r_j}^2j(t)-\tfrac12\lambda_ja_j(t)-\tfrac12\lambda_ja_j(t).
\end{align}
Thus the population equations are
\begin{equation}
\begin{aligned}
 j(t)&=\sum_j\lambda_j a_j(t),\\
 a_j'(t)&=-\lambda_j a_j(t)+\abs{r_j}^2j(t),\\
 f'(t)&=p\,j(t).
 \end{aligned}
 \label{eq:modeflow}
\end{equation}
Population leaves mode $j$ with flux $\lambda_ja_j(t)$. Resets inject target population at rate $pj(t)$ and return the remaining population to the non-target modes with weights $\abs{r_j}^2$ (see also Fig.~\ref{fig:flow}).

\begin{figure}[tbp]
\centering
\resizebox{\linewidth}{!}{%
\begin{tikzpicture}[
 font=\small,
 >={Latex[length=2.6mm,width=1.8mm]},
 block/.style={draw=black!75,rounded corners=2pt,line width=0.7pt,
   fill=white,align=center,minimum height=1.48cm,inner sep=6pt},
 flow/.style={->,line width=0.9pt},
 book/.style={->,densely dashed,line width=0.9pt}]
 \node[block,text width=3.6cm] (modes) at (1.9,0)
   {\textbf{Residual modes}\\[3pt]
    $G\ket{v_j}=\lambda_j\ket{v_j}$\\
    populations $a_j(t)$};
 \node[block,text width=2.5cm,fill=black!3] (reset) at (7.75,0)
   {\textbf{Reset event}\\[3pt]
    prepare $\ket r$\\
    \footnotesize an operation};
 \node[block,text width=3.4cm] (target) at (13.45,0)
   {\textbf{Dark solution state}\\[3pt]
    $P=\ket\psi\bra\psi$\\
    no outgoing jumps};
 \draw[flow] (modes.east) --
   node[above,align=center,font=\footnotesize] {residual detection\\$j(t)=\sum_j\lambda_j a_j(t)$}
   (reset.west);
 \draw[book] (reset.east) --
   node[above,align=center,font=\footnotesize] {target influx\\$p\,j(t)$}
   (target.west);
 \draw[book] (reset.south) -- ++(0,-1.1) -|
   node[pos=0.25,below=3pt,align=center,font=\footnotesize]
   {error reinjection: $(1-p)j(t)$ in total\\mode $j$ receives $|r_j|^2j(t)$}
   (modes.south);
 \node[align=center,font=\footnotesize] at (13.45,-1.55)
   {$f'(t)=p\,j(t)\ge0$\\target population accumulates};
\end{tikzpicture}}
\caption{Population flow for the purely dissipative residual-reset solver. }
\label{fig:flow}
\end{figure}
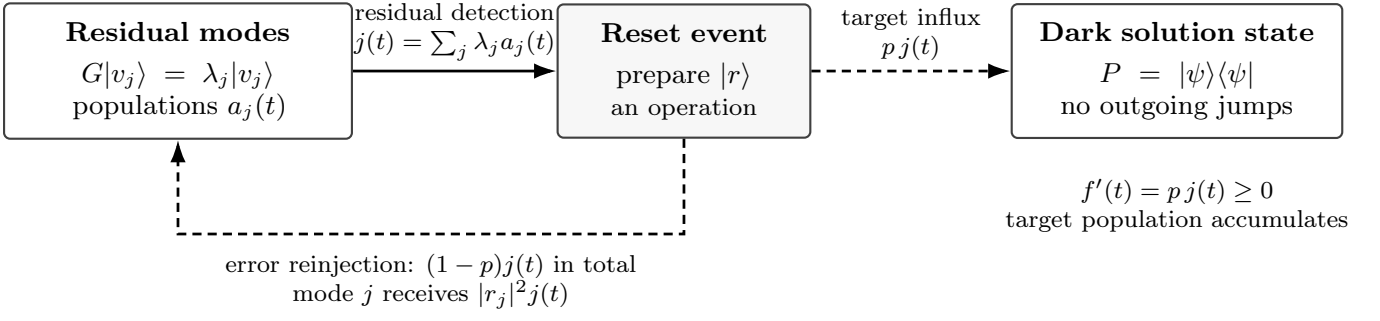

\subsection{Solution extraction and the algorithm}
Let $E_x$ be the projector from $\HS$ onto the original solution coordinates. At the stationary state,
\begin{equation}
 E_x\ket\psi=\frac{s}{\sqrt{s^2+1}}\ket x,
 \qquad a_*:=\Tr(E_xP)=\frac{s^2}{s^2+1}\ge\frac12.
 \label{eq:extraction}
\end{equation}
A two-outcome measurement of $\{E_x,R\}$ therefore extracts exactly $\ket x$ with constant probability. 

The algorithm is now completely specified. Initialize the system at any initial state, evolve under the fixed generator \eqref{eq:lindblad} for a duration $T$ chosen from the accuracy guarantee, and measure whether the system lies in $\HX$ or in the reference level (only a constant expected number of repetitions is needed).

\section{Complexity analysis}\label{sec:complexity}
We analyze three distinct resources in order: the relaxation time of the fixed generator, the cost of accessing its jumps, and the number of input-oracle queries required to simulate a sufficiently long evolution. The final subsection establishes a matching worst-case mixing lower bound for the generator as written.

\subsection{Mixing-time bound}\label{sec:mixing}
\begin{theorem}[Uniform mixing for the pure-reset generator]\label{thm:mix}
Under \eqref{eq:assumptions}, the unique stationary density matrix of \eqref{eq:lindblad} is $P$. For every initial density matrix $\rho_0$ and every $t\ge0$,
\begin{equation}
 D(\ee^{t\Lin}(\rho_0),P)
 \le\sqrt{\Tr(Q\rho_0)}\,
       2^{-\frac12\lfloor t/(3\kappa^2)\rfloor}.
 \label{eq:mixbound}
\end{equation}
Consequently, for $0<\eps\le1/2$,
\begin{equation}
 \tmix(\eps)\le
 3\kappa^2\left\lceil2\log_2\frac1\eps\right\rceil
 =O\!\left(\kappa^2\log\frac1\eps\right).
 \label{eq:mixingtime}
\end{equation}
\end{theorem}
\begin{proof}
\proofstep{Step 1: define the inverse on the non-target subspace.}
Choose singular vectors of $F$ so that
\begin{equation}
 F=\sum_{j=1}^N\sigma_j\ket{u_j}\bra{v_j},\qquad
 \sigma_j\ge\kappa^{-1}.
 \label{eq:proof-svd}
\end{equation}
The $\ket{u_j}$ form an orthonormal basis of $\HR$. The $\ket{v_j}$ form an orthonormal basis of the subspace orthogonal to $\ket\psi$, so that
\begin{equation}
 Q=\sum_j\ket{v_j}\bra{v_j},\qquad
 G=\sum_j\sigma_j^2\ket{v_j}\bra{v_j}.
\end{equation}
Define the Moore--Penrose inverse by inverting only the nonzero eigenvalues:
\begin{equation}
 G^+:=\sum_j\sigma_j^{-2}\ket{v_j}\bra{v_j},\qquad
 G^+\ket\psi=0.
 \label{eq:Gplus-definition}
\end{equation}
Multiplying the spectral expansions and using $\sigma_j^{-2}\le\kappa^2$ gives
\begin{equation}
 GG^+=G^+G=Q,\qquad 0\le G^+\le\kappa^2Q.
 \label{eq:pinverse}
\end{equation}
We also need an expression involving the ordinary inverse of $FF^\dagger$. Since
$(FF^\dagger)^{-2}=\sum_j\sigma_j^{-4}\ket{u_j}\bra{u_j}$, we have
\begin{align}
 F^\dagger(FF^\dagger)^{-2}F
 &=\sum_j\sigma_j\sigma_j^{-4}\sigma_j\ket{v_j}\bra{v_j}\notag\\
 &=\sum_j\sigma_j^{-2}\ket{v_j}\bra{v_j}=G^+.
 \label{eq:Gplus-viaF}
\end{align}
Here $G$ is singular, whereas $FF^\dagger$ is invertible. The two inverses in these formulas therefore have different meanings.

\proofstep{Step 2: calculate the inverse-weighted reset overlap.}
Set $B=AA^\dagger$, a positive-definite operator on $\HR$. First,
\begin{equation}
 B^{-1}=(A^{-1})^\dagger A^{-1},\qquad
 \bra bB^{-1}\ket b=\norm{A^{-1}\ket b}^{\,2}=s^2.
 \label{eq:Bnormidentity}
\end{equation}
To find $(FF^\dagger)^{-1}\ket b$, multiply $B^{-1}\ket b$ by $FF^\dagger=B+\ket b\bra b$:
\begin{align}
 (FF^\dagger)B^{-1}\ket b
 &=(B+\ket b\bra b)B^{-1}\ket b\notag\\
 &=\ket b+\ket b\bigl(\bra bB^{-1}\ket b\bigr)
 =(1+s^2)\ket b.
 \label{eq:rankone-multiplication}
\end{align}
Invertibility of $FF^\dagger$ therefore implies
\begin{equation}
 (FF^\dagger)^{-1}\ket b=\frac{B^{-1}\ket b}{1+s^2}.
 \label{eq:rankone}
\end{equation}
This is an identity for the action on $\ket b$, not an equality of the two inverse operators up to a scalar.

Next, use \eqref{eq:Gplus-viaF} and $F\ket r=-\ket b$. The two minus signs cancel, giving
\begin{align}
 \bra rG^+\ket r
 &=\bra rF^\dagger(FF^\dagger)^{-2}F\ket r\notag\\
 &=\bra b(FF^\dagger)^{-2}\ket b
 =\norm{(FF^\dagger)^{-1}\ket b}^{\,2}\notag\\
 &=\frac{\norm{B^{-1}\ket b}^{\,2}}{(1+s^2)^2}
 =\frac{\bra bB^{-2}\ket b}{(1+s^2)^2}.
 \label{eq:reset-inverse-overlap}
\end{align}
Both squared-norm equalities use Hermiticity of the corresponding positive inverse. In particular, the squared denominator comes from taking the norm squared in \eqref{eq:rankone}.

Define
\begin{equation}
 \beta:=\frac{\bra rG^+\ket r}{p}
 =\frac{\bra bB^{-2}\ket b}{1+s^2},
 \label{eq:beta-definition}
\end{equation}
since $p=(1+s^2)^{-1}$. Every eigenvalue $\mu$ of $B$ satisfies $\mu\ge\kappa^{-2}$, and hence $\mu^{-2}\le\kappa^2\mu^{-1}$. Applying this inequality in an eigenbasis of $B$ proves $B^{-2}\le\kappa^2B^{-1}$. Therefore
\begin{equation}
 \beta
 \le\frac{\kappa^2\bra bB^{-1}\ket b}{1+s^2}
 =\kappa^2\frac{s^2}{1+s^2}\le\kappa^2.
 \label{eq:beta}
\end{equation}

\proofstep{Step 3: derive two scalar evolution equations.}
For $\rho_t=\ee^{t\Lin}(\rho_0)$, use the populations and jump intensity from Section~\ref{sec:construction}, and introduce one additional expectation:
\begin{equation}
 \begin{aligned}
 f(t)&=\Tr(P\rho_t),& q(t)&=1-f(t),\\
 j(t)&=\Tr(G\rho_t),& u(t)&=\Tr(G^+\rho_t).
 \end{aligned}
 \label{eq:proof-scalars}
\end{equation}
Equation~\eqref{eq:fidelity} gives $f'(t)=p\,j(t)\ge0$, so $q$ is nonincreasing. Differentiating $u$, inserting the generator, and cyclically moving the last $G$ inside the trace gives
\begin{align}
 u'(t)
 &=j(t)\Tr(G^+R)
   -\tfrac12\Tr(G^+G\rho_t)-\tfrac12\Tr(G^+\rho_tG)\notag\\
 &=j(t)\bra rG^+\ket r
   -\tfrac12\Tr\bigl((G^+G+GG^+)\rho_t\bigr)\notag\\
 &=p\beta j(t)-\Tr(Q\rho_t)
 =\beta f'(t)-q(t).
 \label{eq:u-derivative}
\end{align}
Taking expectations of $0\le G^+\le\kappa^2Q$ yields
\begin{equation}
 u'(t)=\beta f'(t)-q(t),\qquad
 0\le u(t)\le\kappa^2q(t).
 \label{eq:bookkeeping}
\end{equation}

\proofstep{Step 4: bound the fidelity gain over a finite interval.}
Fix a starting time $t\ge0$ and an interval length $T>0$. Write
$\Delta f=f(t+T)-f(t)\ge0$, so $q(t+T)=q(t)-\Delta f$.
Integrating \eqref{eq:u-derivative} and rearranging gives
\begin{align}
 \int_t^{t+T}q(v)\,\dd v
 &=u(t)-u(t+T)+\beta\Delta f\notag\\
 &\le\kappa^2q(t)+\kappa^2\Delta f.
 \label{eq:integrated-proof}
\end{align}
The inequality uses $u(t)\le\kappa^2q(t)$, $u(t+T)\ge0$, and $\beta\le\kappa^2$ separately. On the same interval, monotonicity gives $q(v)\ge q(t+T)$, and hence
\begin{equation}
 \int_t^{t+T}q(v)\,\dd v
 \ge Tq(t+T)=T[\,q(t)-\Delta f\,].
 \label{eq:integral-lower}
\end{equation}
Combining the two bounds and collecting the terms containing $\Delta f$ gives
\begin{equation}
 (T-\kappa^2)q(t)\le(T+\kappa^2)\Delta f.
\end{equation}
Thus, for $T\ge\kappa^2$,
\begin{equation}
 f(t+T)-f(t)\ge
 \frac{T-\kappa^2}{T+\kappa^2}[1-f(t)].
 \label{eq:finitegain}
\end{equation}
At $T=3\kappa^2$ the fraction is $1/2$. Equivalently,
\begin{equation}
 q(t+3\kappa^2)\le\frac12q(t).
 \label{eq:halving}
\end{equation}

\proofstep{Step 5: iterate and convert population error to trace distance.}
For any $t\ge0$, let $n=\lfloor t/(3\kappa^2)\rfloor$. Repeated application of \eqref{eq:halving}, followed by monotonicity on the remaining part of the interval, gives
\begin{equation}
 q(t)\le q(3n\kappa^2)\le2^{-n}q(0).
 \label{eq:qbound}
\end{equation}
Based on the relations
\begin{equation}
 1-\Tr(P_*\rho)\le D(\rho,P_*)
  \le\sqrt{1-\Tr(P_*\rho)},
 \label{eq:pure-distance}
\end{equation}
we obtain
\begin{equation}
 D(\rho_t,P)\le\sqrt{q(t)}
 \le\sqrt{q(0)}\,2^{-n/2},
\end{equation}
which is \eqref{eq:mixbound}. Since $q(0)\le1$, choosing $n\ge\lceil2\log_2(1/\eps)\rceil$ proves \eqref{eq:mixingtime}.
Finally, $P$ is stationary by \eqref{eq:fixedpoint}. If another density matrix $\rho_*$ were stationary, then $\ee^{t\Lin}(\rho_*)=\rho_*$ for all $t$. Applying \eqref{eq:mixbound} to $\rho_*$ and sending $t\to\infty$ would give $D(\rho_*,P)=0$. Thus $\rho_*=P$, proving uniqueness.
\end{proof}

\subsection{Operational construction of the jump operators}\label{sec:implementation}
We first construct the residual block encoding from $U_A$ and $U_b$, and then convert it into collective jump access. All sector tests and basis-state swaps below concern known labels. No projector onto the unknown solution is used.

\paragraph{Residual block encoding from the input oracles.}
Identify $\HR$ and $\HX$ by their computational bases, and let $\iota:\HX\to\HS$ be the embedding. 
\begin{equation}
 \begin{aligned}
 \iota&=\sum_{j=1}^N\ket j_S\bra j_x
       =\begin{pmatrix}I_N\\0_{1\times N}\end{pmatrix},\\
 \iota^\dagger&=\begin{pmatrix}I_N&0_{N\times1}\end{pmatrix},\qquad
 \iota^\dagger\iota=I_x.
 \end{aligned}
 \label{eq:embedding}
\end{equation}
In the ordered system basis $\{\ket1,\ldots,\ket N,\ket r\}$, define
\begin{equation}
 \begin{aligned}
  \bar A&=\iota A\iota^\dagger
  =\begin{pmatrix}A&0_{N\times1}\\0_{1\times N}&0\end{pmatrix},
  \\[1ex]
  \ket{\bar b}&=\iota\ket b=\begin{pmatrix}\ket b\\0\end{pmatrix},\\[1ex]
  \widehat F&=\iota F=\bar A-\ket{\bar b}\bra r
  =\begin{pmatrix}A&-\ket b\\0_{1\times N}&0\end{pmatrix}.
 \end{aligned}
 \label{eq:paddedF}
\end{equation}
Here $\bar A$ acts as $A$ on the solution sector and as zero on $\ket r$, while $\ket{\bar b}$ is $\ket b$ with a zero reference coordinate appended; the bars denote these embeddings, not complex conjugation. Thus $\widehat F$ is a square operator on the augmented register whose output has zero reference component. 

\begin{proposition}[Explicit residual block encoding]\label{prop:Fencoding}
Assume the input oracles satisfy
\begin{equation}
 \begin{gathered}
 (\bra{0^a}\otimes I_x)U_A(\ket{0^a}\otimes I_x)=A,\\
 U_b\ket0=\ket b.
 \end{gathered}
\end{equation}
Use a selector qubit $c$, a flag qubit $f$, and the $a$ matrix-oracle signal qubits, so that $K=(c,f,a)$ contains $m=a+2$ signal qubits. Let $R=\ket r\bra r$, $E_R=I_S-R$, and let $T_{r,0_x}$ exchange the two known basis states $\ket r$ and $\ket{0_x}=\iota\ket0$, leaving their orthogonal complement unchanged. Define
\begin{equation}
 \begin{aligned}
 \widetilde U_A&=
 \begin{pmatrix}U_A&0\\0&I_a\end{pmatrix},
 \\[1ex]
 W_b&=
 \begin{pmatrix}U_b&0\\0&1\end{pmatrix}T_{r,0_x},
 \\[1ex]
 C_{fS}&=I_f\otimes E_R+X_f\otimes R.
 \end{aligned}
 \label{eq:UF-gates}
\end{equation}
The displayed blocks separate the solution and reference sectors; $\widetilde U_A$ also includes the oracle signal register. Thus $\widetilde U_A$ applies $U_A$ in the solution sector and acts as the identity in the reference sector. The known gate $C_{fS}$ flips $f$ exactly on the reference level. With identities on untouched registers implicit, set
\begin{equation}
\begin{aligned}
 U_F={}&(H_c\otimes I)\\[-1mm]
 &\times\Bigl[\ket0\bra0_c\otimes(I_f\otimes\widetilde U_A)\\[-1mm]
 &\qquad-\ket1\bra1_c\otimes(X_f\otimes I_a\otimes W_b)\Bigr]\\[-1mm]
 &\times(H_c\otimes I)\,C_{fS},
 \end{aligned}
 \label{eq:select}
\end{equation}
where $H_c$ is the Hadamard gate and $X_f$ is the Pauli $X$ gate. Then $U_F$ is unitary and
\begin{equation}
 \begin{aligned}
 &(\bra{0^m}_K\otimes I_S)U_F(\ket{0^m}_K\otimes I_S)\\
 &\qquad=\frac{\widehat F}{2}.
 \end{aligned}
 \label{eq:Fblock}
\end{equation}
One call to $U_F$, its inverse, or a controlled version uses $O(1)$ queries to $U_A,U_b$ and their inverses.
\end{proposition}
\begin{proof}
Let $\ket\phi_S$ be arbitrary, and write
$\ket{uv}_{cf}:=\ket u_c\otimes\ket v_f$.
We apply the factors of $U_F$ in \eqref{eq:select} from right to left,
starting with
\begin{equation}
 \ket{0^m}_K\otimes\ket\phi_S
 =
 \ket{00}_{cf}\otimes\ket{0^a}_a\otimes\ket\phi_S.
\end{equation}

First, applying
$C_{fS}=I_f\otimes E_R+X_f\otimes R$ gives
\begin{equation}
 \begin{aligned}
 &C_{fS}\bigl(\ket{0^m}_K\otimes\ket\phi_S\bigr)\\
 &\quad=
 \ket{00}_{cf}\otimes\ket{0^a}_a\otimes E_R\ket\phi_S
 +
 \ket{01}_{cf}\otimes\ket{0^a}_a\otimes R\ket\phi_S.
 \end{aligned}
\end{equation}

Next, the first Hadamard on $c$ transforms this state into
\begin{equation}
 \frac{1}{\sqrt2}
 \begin{aligned}[t]
 \Bigl[
 &\bigl(\ket{00}_{cf}+\ket{10}_{cf}\bigr)
   \otimes\ket{0^a}_a\otimes E_R\ket\phi_S\\
 +&\bigl(\ket{01}_{cf}+\ket{11}_{cf}\bigr)
   \otimes\ket{0^a}_a\otimes R\ket\phi_S
 \Bigr].
 \end{aligned}
\end{equation}

The selector-controlled operator applies $\widetilde U_A$ when $c=0$,
and applies $-X_f\otimes I_a\otimes W_b$ when $c=1$.
The resulting state is therefore
\begin{equation}
 \frac{1}{\sqrt2}
 \begin{aligned}[t]
 \Bigl[
 &\ket{00}_{cf}\otimes
   \widetilde U_A
   \bigl(\ket{0^a}_a\otimes E_R\ket\phi_S\bigr)\\
 +&\ket{01}_{cf}\otimes
   \widetilde U_A
   \bigl(\ket{0^a}_a\otimes R\ket\phi_S\bigr)\\
 -&\ket{11}_{cf}\otimes\ket{0^a}_a
   \otimes W_bE_R\ket\phi_S\\
 -&\ket{10}_{cf}\otimes\ket{0^a}_a
   \otimes W_bR\ket\phi_S
 \Bigr].
 \end{aligned}
\end{equation}

Now apply the final Hadamard and project $c,f$ onto $\ket{00}_{cf}$.
Since
\begin{equation}
 \bra{00}_{cf}(H_c\otimes I_f)
 =
 \frac{\bra{00}_{cf}+\bra{10}_{cf}}{\sqrt2},
\end{equation}
only the first and fourth terms contribute. Hence
\begin{equation}
 \begin{aligned}
 &(\bra{00}_{cf}\otimes I_{aS})
 U_F\bigl(\ket{0^m}_K\otimes\ket\phi_S\bigr)\\
 &\quad=
 \frac12\Bigl[
 \widetilde U_A
 \bigl(\ket{0^a}_a\otimes E_R\ket\phi_S\bigr)
 -
 \ket{0^a}_a\otimes W_bR\ket\phi_S
 \Bigr].
 \end{aligned}
\end{equation}

Finally, project the matrix-oracle signal register onto $\ket{0^a}_a$.
The definition of $\widetilde U_A$ gives
\begin{equation}
 \begin{aligned}
 &(\bra{0^a}_a\otimes I_S)\widetilde U_A
   (\ket{0^a}_a\otimes I_S)\\
 &\qquad=
 \begin{pmatrix}
 A&0\\
 0&1
 \end{pmatrix}
 =\bar A+R.
 \end{aligned}
\end{equation}
For the source term,
\begin{align}
 W_bR
 &=
 \begin{pmatrix}
 U_b&0\\
 0&1
 \end{pmatrix}
 T_{r,0_x}\ket r\bra r
 \notag\\
 &=
 \begin{pmatrix}
 U_b&0\\
 0&1
 \end{pmatrix}
 \ket{0_x}\bra r
 \notag\\
 &=\ket{\bar b}\bra r.
\end{align}
Consequently,
\begin{align}
 &(\bra{0^m}_K\otimes I_S)
 U_F\bigl(\ket{0^m}_K\otimes\ket\phi_S\bigr)
 \notag\\
 &\qquad=
 \frac12\Bigl[(\bar A+R)E_R-W_bR\Bigr]\ket\phi_S
 \notag\\
 &\qquad=
 \frac12\Bigl[\bar A-\ket{\bar b}\bra r\Bigr]\ket\phi_S
 \notag\\
 &\qquad=
 \frac12
 \begin{pmatrix}
 A&-\ket b\\
 0&0
 \end{pmatrix}
 \ket\phi_S
 \notag\\
 &\qquad=\frac{\widehat F}{2}\ket\phi_S,
\end{align}
where we used $RE_R=0$ and $\bar A E_R=\bar A$.
Since $\ket\phi_S$ is arbitrary, this proves
\begin{equation}
 (\bra{0^m}_K\otimes I_S)
 U_F
 (\ket{0^m}_K\otimes I_S)
 =
 \frac{\widehat F}{2}.
\end{equation}

All factors of $U_F$ are unitary.
Indeed, $\widetilde U_A$ is a direct sum of unitaries, and $W_b$
is a product of unitaries. The orthogonal-projector identities
$E_RR=0$ and $E_R+R=I_S$ imply
$C_{fS}^{\dagger}C_{fS}=I$.
The selector-controlled operator is unitary on each of its two
orthogonal selector subspaces, and the two Hadamards are unitary.
Thus their product $U_F$ is unitary.

The circuit uses one controlled query to $U_A$ and one controlled
query to $U_b$. All remaining operations act on known flags or basis
labels. Reversing the circuit or adding an external control therefore
preserves the $O(1)$ input-query count.
\end{proof}

\paragraph{Collective jump encoding by a register SWAP.}
The collective jump map and its padded version are
\begin{equation}
 \begin{gathered}
 V_J=\sum_{j=1}^N\ket j_E\otimes J_j:
 \HS\longrightarrow\mathcal H_{\mathrm{res},E}\otimes\HS,\\
 \widehat V_J=(\iota_E\otimes I_S)V_J,
 \end{gathered}
 \label{eq:collective}
\end{equation}
where $E$ in the padded implementation is a copy of the augmented system register. 

\begin{proposition}[SWAP construction of collective jump access]\label{prop:swapencoding}
Let $U_F$ satisfy \eqref{eq:Fblock}, initialize $E$ in $\ket r_E$, and set
\begin{equation}
 U_J=(I_K\otimes\operatorname{SWAP}_{ES})\,U_{F,KS},
 \label{eq:UJ}
\end{equation}
where $U_{F,KS}$ acts trivially on $E$. Then
\begin{equation}
\begin{aligned}
 &(\bra{0^m}_K\otimes I_{ES})\,U_J\\
 &\quad\times(\ket{0^m}_K\otimes\ket r_E\otimes I_S)
 =\frac{\widehat V_J}{2}.
 \end{aligned}
 \label{eq:VJblock}
\end{equation}
\end{proposition}
\begin{proof}
\proofstep{Step 1: identify the desired joint output.}
Fix a normalized system input $\ket\phi$ and define its residual components
\begin{equation}
 c_j:=\bra jF\ket\phi,\qquad
 F\ket\phi=\sum_{j=1}^N c_j\ket j.
 \label{eq:residual-components}
\end{equation}
Each $c_j$ is a scalar. From $J_j=\ket r\bra jF$ we have $J_j\ket\phi=c_j\ket r$, and hence
\begin{align}
 \widehat V_J\ket\phi
 &=\sum_{j=1}^N\ket j_E\otimes(J_j\ket\phi)_S\notag\\
 &=\sum_{j=1}^N\ket j_E\otimes(c_j\ket r_S)\notag\\
 &=\left(\sum_{j=1}^N c_j\ket j_E\right)\otimes\ket r_S\notag\\
 &=(\widehat F\ket\phi)_E\otimes\ket r_S.
 \label{eq:VJ-identity}
\end{align}
The residual basis in $E$ is embedded by $\iota_E$; the additional reference coordinate of $\widehat F\ket\phi$ is zero. Factoring out the common system ket $\ket r_S$ is the only step relating the collective-jump notation to the residual-output notation.

\proofstep{Step 2: apply the residual encoding and isolate its signal component.}
Initialize the registers in the order $K,E,S$ as
$\ket{0^m}_K\otimes\ket r_E\otimes\ket\phi_S$.
Proposition~\ref{prop:Fencoding} states that projection of the $U_F$ output onto the zero-signal ancillas gives $\widehat F\ket\phi/2$. Equivalently, the complete unitary output decomposes as
\begin{equation}
 \begin{split}
 &U_{F,KS}(\ket{0^m}_K\otimes\ket r_E\otimes\ket\phi_S)\\
 &\quad=\tfrac12\ket{0^m}_K\otimes\ket r_E
       \otimes(\widehat F\ket\phi)_S\\
 &\qquad+\uket{\mathrm{fail}_\phi},
 \end{split}
 \label{eq:UF-branch}
\end{equation}
where the generally unnormalized remainder satisfies
\begin{equation}
 (\bra{0^m}_K\otimes I_{ES})\uket{\mathrm{fail}_\phi}=0.
 \label{eq:failure-orthogonality}
\end{equation}

\proofstep{Step 3: exchange the system and environment registers.}
The registers $S$ and $E$ have equal padded dimensions, so their SWAP is unitary. It maps
$\ket r_E\otimes(\widehat F\ket\phi)_S$ to
$(\widehat F\ket\phi)_E\otimes\ket r_S$.
Consequently,
\begin{equation}
 \begin{split}
 &U_J(\ket{0^m}_K\otimes\ket r_E\otimes\ket\phi_S)\\
 &\quad=\tfrac12\ket{0^m}_K\otimes(\widehat F\ket\phi)_E
       \otimes\ket r_S\\
 &\qquad+\uket{\mathrm{fail}'_\phi}.
 \end{split}
 \label{eq:UJ-branch}
\end{equation}
Since SWAP does not act on $K$, it commutes with the zero-signal projector and preserves \eqref{eq:failure-orthogonality}. The first term, by \eqref{eq:VJ-identity}, is exactly
\begin{equation}
 \tfrac12\ket{0^m}_K\otimes(\widehat V_J\ket\phi)_{ES}.
\end{equation}
Projecting the signal register onto zero therefore gives $\widehat V_J\ket\phi/2$. This holds for every $\ket\phi$, proving the operator identity \eqref{eq:VJblock}.

Finally, $U_J$ consists of one $U_F$ call and a register SWAP. The inverse applies SWAP first and then $U_F^\dagger$. Both, and their controlled versions, use $O(1)$ input queries. 
\end{proof}

\subsection{From mixing time to query complexity}\label{sec:queries}
The collective-encoding simulation theorem of Chen, Kastoryano, Brand\~ao, and Gily\'en~\cite[Theorem III.2 and Appendix F]{CKBG} simulates a unit-normalized purely dissipative generator for time $T$ to diamond-norm error $\eta$ using
\begin{equation}
 O\!\left((T+1)\frac{\log((T+1)/\eta)}{\log\log((T+1)/\eta)}\right)
 \label{eq:simulation}
\end{equation}
controlled queries to its collective encoding and inverse. The theorem is used as a simulation subroutine; we do not reprove it here. 

For our input oracle, the encoded jumps are $\overline J_j=J_j/2$. Every dissipative term is quadratic in its jump, so their generator is
\begin{equation}
 \begin{aligned}
 \overline\Lin(\rho)
 &=\sum_j\left(\overline J_j\rho\overline J_j^\dagger
 -\tfrac12\{\overline J_j^\dagger\overline J_j,\rho\}\right)=\tfrac14\Lin(\rho).
 \end{aligned}
 \label{eq:normalized-generator}
\end{equation}
To realize the physical evolution time $T$, the normalized simulation time is $4T$, a constant factor.

\begin{corollary}[Heralded solution preparation]\label{cor:queries}
For $0<\eps<1/4$, under \eqref{eq:oracles}, the pure-reset solver prepares a heralded output within trace distance $\eps$ of $\ket x\bra x$, with constant success probability and expected input-query complexity
\begin{equation}
 O\!\left(\kappa^2\log(1/\eps)
             \frac{\log(\kappa/\eps)}{\log\log(\kappa/\eps)}\right).
 \label{eq:queries}
\end{equation}
\end{corollary}
\begin{proof}
\proofstep{Step 1: choose the evolution and simulation errors.}
Initialize at $R$ and choose
\begin{equation}
 T=3\kappa^2\left\lceil2\log_2\frac{16}{\eps}\right\rceil,
 \qquad \eta=\frac\eps8.
 \label{eq:budget}
\end{equation}
Theorem~\ref{thm:mix} gives $D(\ee^{T\Lin}(R),P)\le\eps/16$.
Let $\widetilde{\mathcal E}_T$ be the simulated channel with
$\norm{\widetilde{\mathcal E}_T-\ee^{T\Lin}}_\diamond\le\eta$, and write
$\widehat\rho=\widetilde{\mathcal E}_T(R)$. By the definition of diamond norm,
\begin{equation}
 \begin{aligned}
 D(\widehat\rho,\ee^{T\Lin}(R))
 &=\tfrac12\norm{(\widetilde{\mathcal E}_T-\ee^{T\Lin})(R)}_1\\
 &\le\frac\eta2=\frac\eps{16}.
 \end{aligned}
\end{equation}
The triangle inequality then gives
\begin{equation}
 D(\widehat\rho,P)\le\frac\eps8=:d.
 \label{eq:prepared-error}
\end{equation}

\proofstep{Step 2: bound the probability of extracting the solution.}
Let $a=\Tr(E_x\widehat\rho)$ and $a_*=\Tr(E_xP)\ge1/2$.
The two-outcome measurement $\{E_x,R\}$ cannot increase trace distance, so
$\abs{a-a_*}\le D(\widehat\rho,P)\le d$. Therefore, since $\eps<1/4$,
\begin{equation}
 a\ge a_*-d\ge\frac12-d\ge\frac{15}{32}.
 \label{eq:actual-success}
\end{equation}

\proofstep{Step 3: control the normalized conditional state.}
Identify the solution sector with $\HX$, and write
$\widehat\rho_x=E_x\widehat\rho E_x/a$.
At the ideal target, \eqref{eq:extraction} implies
$E_xPE_x=a_*\ket x\bra x$. Subtracting the two conditional states gives the exact identity
\begin{equation}
 \begin{aligned}
 &a(\widehat\rho_x-\ket x\bra x)\\
 &\quad=E_x(\widehat\rho-P)E_x+(a_*-a)\ket x\bra x.
 \end{aligned}
 \label{eq:conditional-difference}
\end{equation}
The compression bound
$\norm{E_x(\widehat\rho-P)E_x}_1\le\norm{\widehat\rho-P}_1\le2d$
follows from $\norm{E_x}=1$. Also $\norm{\ket x\bra x}_1=1$ and $\abs{a-a_*}\le d\le2d$. Taking trace norms in \eqref{eq:conditional-difference} thus yields
\begin{align}
 D(\widehat\rho_x,\ket x\bra x)
 &\le\frac{2d+\abs{a-a_*}}{2a}
 \le\frac{2d}{a}\notag\\
 &\le\frac{2d}{1/2-d}
 \le\frac{8\eps}{15}<\eps.
 \label{eq:conditioning}
\end{align}
The loose middle inequality keeps the constants simple; no assumption that postselection preserves the original trace distance is used.

\proofstep{Step 4: substitute the evolution time into the query bound.}
Propositions~\ref{prop:Fencoding} and~\ref{prop:swapencoding} show that one call to the normalized collective oracle or its inverse costs $O(1)$ input queries. Equation~\eqref{eq:normalized-generator} requires normalized simulation time $4T$. Set $\ell=\log(1/\eps)$. For $\kappa\ge1$ and $0<\eps<1/4$,
\begin{equation}
 \begin{aligned}
 T&=\Theta(\kappa^2\ell),\\
 \log\frac{4T+1}{\eta}
 &=\Theta\bigl(\log(\kappa/\eps)\bigr).
 \end{aligned}
 \label{eq:log-substitution}
\end{equation}
Indeed, the logarithm on the left is $2\log\kappa+\ell+\log\ell+O(1)$, and $\ell\ge\log4$ makes $\log\ell$ at most a constant multiple of $\ell$. Inserting these estimates into \eqref{eq:simulation} gives the bound \eqref{eq:queries} for one attempt.

\proofstep{Step 5: account for restarting.}
Each fresh attempt uses the same simulated channel and succeeds with probability $a\ge15/32$. Independent restarts therefore have a geometric number of attempts with mean $1/a\le32/15$. Conditioning on the first successful attempt gives the same state $\widehat\rho_x$ as conditioning one attempt on success. Thus restarting multiplies the expected query cost by a constant and preserves \eqref{eq:conditioning}.
\end{proof}

\subsection{A matching lower bound and its scope}\label{sec:lower}
\begin{proposition}[Sharp conditioning and precision dependence]\label{prop:lower}
For every $\kappa\ge1$, there is an input obeying \eqref{eq:assumptions} for which the generator \eqref{eq:lindblad} satisfies
\begin{equation}
 \tmix(\eps)\ge\kappa^2\log\frac1\eps,
 \qquad 0<\eps<1.
 \label{eq:lower}
\end{equation}
\end{proposition}
\begin{proof}
Use the input and initialization
\begin{equation}
 \begin{gathered}
 A=\diag(1,\kappa^{-1}),\\
 \ket b=\ket1,\qquad \rho_0=\ket2\bra2.
 \end{gathered}
 \label{eq:lowerinstance}
\end{equation}
The singular values of $A$ obey \eqref{eq:assumptions}. Since $A^{-1}\ket b=\ket1$, the target is
$\ket\psi=(\ket1+\ket r)/\sqrt2$. In the system basis $(\ket1,\ket2,\ket r)$,
\begin{equation}
 \begin{gathered}
 F=\begin{pmatrix}1&0&-1\\0&\kappa^{-1}&0\end{pmatrix},\\[1ex]
 G=\begin{pmatrix}1&0&-1\\0&\kappa^{-2}&0\\-1&0&1\end{pmatrix}.
 \end{gathered}
 \label{eq:lower-matrices}
\end{equation}
Thus $G\ket2=\kappa^{-2}\ket2$, whereas
$\bra2R\ket2=\bra2P\ket2=0$.
Let $z(t)=\bra2\rho_t\ket2$. Taking this matrix element of the full generator yields
\begin{align}
 z'(t)
 &=\Tr(G\rho_t)\bra2R\ket2\notag\\
 &\quad-\tfrac12\bra2G\rho_t\ket2-\tfrac12\bra2\rho_tG\ket2\notag\\
 &=0-\tfrac12\kappa^{-2}z(t)-\tfrac12\kappa^{-2}z(t)\notag\\
 &=-\kappa^{-2}z(t).
 \label{eq:slow-population-ode}
\end{align}
Since $z(0)=1$, the exact solution is
\begin{equation}
 z(t)=\bra2\rho_t\ket2=\ee^{-t/\kappa^2}.
 \label{eq:slowpopulation}
\end{equation}
The measurement $\{\ket2\bra2,I-\ket2\bra2\}$ has outcome probabilities $(z(t),1-z(t))$ on $\rho_t$ and $(0,1)$ on $P$. Their classical total-variation distance is $z(t)$, hence $D(\rho_t,P)\ge z(t)$. Any time satisfying the worst-case mixing requirement must therefore obey
$\ee^{-t/\kappa^2}\le\eps$, or $t\ge\kappa^2\log(1/\eps)$.
\end{proof}

Together, Theorem~\ref{thm:mix} and Proposition~\ref{prop:lower} establish the worst-case scaling
\begin{equation}
\sup_{A,\ket b}\tmix(\eps)
 =\Theta\!\left(\kappa^2\log\frac1\eps\right)
 \label{eq:tight}
\end{equation}
for this rate-normalized residual-reset family, uniformly for $0<\eps\le1/2$. The lower bound is not an oracle lower bound for linear-system solving, nor a lower bound for every purely dissipative design. Its input has an easily prepared solution, and initialization at $\ket r$ does not populate the slow mode used in the proof. The statement concerns worst-case mixing over all initial states.

\section{Summary and outlook}\label{sec:discussion}
We have shown how a linear-system constraint can directly define a purely dissipative state-preparation algorithm. The residual map $F=(A,-\ket b)$ has a one-dimensional kernel encoding the augmented solution, and the jumps $J_j=\ket r\bra jF$ detect residual amplitude and reset it to a known reference level. The resulting time-independent generator has the solution encoding as its unique attractive pure stationary state. It uses $N+1$ active levels, contains no inverse matrix or target oracle, and allows constant-probability extraction without solution-norm estimation. The $O(\kappa^2\log(1/\eps))$ trace-distance mixing bound is dimension independent, holds for arbitrary initial density matrices, and is tight for this normalized family. Collective jump block encoding construction then connects this dynamical statement to an explicit query complexity.

The construction here aims to show the power of dissipation as a primitive for designing quantum algorithms. Optimal coherent algorithms have better conditioning dependence ($O(\kappa\log(1/\eps))$)~\cite{Costa,Dalzell}, leading to the question whether the Lindbladian approach can be further improved.

\par\bigskip

\textbf{Acknowledgments.}
Z.S. would like to thank Dong An for inspiring discussions. Z.S. acknowledges financial support from the ERC grant GIFNEQ 101163938.

\textbf{AI-use statement.} The initial question of whether dissipation can solve linear systems is raised by human, and is an open question in \cite{ShangODE}. The Lindbladian construction and the complexity analysis were developed through multiple rounds of discussion between the author and the GPT6 Astra, where AI makes the main technical contributions. The author has carefully verified the proofs and substantially rewrote and reorganized the presentation.

\bibliography{references}
\end{document}